\documentclass[3p]{elsarticle} 

\RequirePackage[T1]{fontenc}
\RequirePackage{lmodern}

\usepackage{microtype}

\usepackage{hyperref} 
\usepackage{amsfonts}
\usepackage{amsmath}
\usepackage{enumerate}

\usepackage{array}
\usepackage{microtype}
\usepackage{pgf,tikz}
\usepackage{pgfplots}
\usetikzlibrary{arrows,decorations.pathmorphing,calc,patterns,decorations.pathreplacing,fadings}
\usepgfplotslibrary{fillbetween}
\usepackage{algorithm}
\usepackage[noend]{algpseudocode}

\usepackage{xspace}
\usepackage{graphicx}
\usepackage{paralist}
\usepackage{amssymb}
 
\usepackage{amsthm}

\usepackage{subcaption}

\newcommand{\alg}{\ensuremath{{\mathtt A}}\xspace}
\newcommand{\opt}{\ensuremath{{\mathtt{OPT}}}\xspace}
\newcommand{\compr}{\ensuremath{{\rho}}\xspace}

\def\epsilon{\varepsilon}

\newcommand{\R}{\mathbf{R}}

\newcommand{\RKP}{\text{\textsc{ReserveKP}}\xspace}
\newcommand{\OKP}{\text{\textsc{OnlineKP}}\xspace}

\newcommand{\fourfourfive}{\ensuremath{\alpha_4}\xspace}

\newcommand{\eg}{e.\,g.}
\newcommand{\ie}{i.\,e.}

\newtheorem{theorem}{Theorem}
\newtheorem{lemma}[theorem]{Lemma}

\journal{Journal of Computer and System Sciences}

\begin{document}

\begin{frontmatter}

\title{Online Simple Knapsack with Reservation Costs}

\author[author1]{Hans-Joachim B\"{o}ckenhauer}
\affiliation[author2]{organization={Department of Computer Science, 
ETH Z\"{u}rich},
            addressline={Universit\"{a}tstrasse 6, 8092 Z\"{u}rich}, 
            country={Switzerland}}
\ead{hjb@inf.ethz.ch}

\author[author2]{Elisabet Burjons}
\affiliation[author1]{organization={Department of Computer Science, 
RWTH Aachen University},
            addressline={Ahornstrasse 55, 52074 Aachen},
            country={Germany}}
\ead{burjons@cs.rwth-aachen.de}

\author[author3]{Fabian Frei\corref{cor3}\fnref{fn1}}
\affiliation[author3]{organization={Department of Computer Science, 
ETH Z\"{u}rich},
            addressline={Universit\"{a}tstrasse 6, 8092 
            Z\"{u}rich}, 
            country={Switzerland}}
\ead{fabian.frei@inf.ethz.ch}
\cortext[cor3]{Corresponding author}

\author[author4]{Juraj Hromkovi\v{c}}
\affiliation[author4]{organization={Department of Computer Science, 
ETH Z\"{u}rich},
            addressline={Universit\"{a}tstrasse 6, 8092 
            Z\"{u}rich}, 
            country={Switzerland}}
\ead{juraj.hromkovic@inf.ethz.ch}

\author[author5]{Henri Lotze}
\affiliation[author5]{organization={Department of Computer Science, 
RWTH Aachen University},
            addressline={Ahornstrasse 55, 52074 Aachen},
            country={Germany}}
\ead{lotze@cs.rwth-aachen.de}

\author[author6]{Peter Rossmanith}
\affiliation[author6]{organization={Department of Computer Science, 
RWTH Aachen University},
            addressline={Ahornstrasse 55, 52074 Aachen},
            country={Germany}}
\ead{rossmani@cs.rwth-aachen.de}

\fntext[f1]{A preliminary version of parts of this paper has appeared in the Proceedings of the 38th International Symposium on Theoretical Aspects of Computer Science (STACS 2021)~\cite{BBHLR21}.}

\begin{abstract}
In the \emph{online simple knapsack problem} we are given a knapsack of unit size 1.
Items of size smaller or equal to 1 are presented in an iterative fashion and 
an algorithm has to decide for each item 
whether to reject or permanently include it into the knapsack without any
knowledge about the rest of the instance. The goal is to pack the knapsack
as full as possible. In this work, we introduce a third option additional to
those of packing and rejecting an item, namely that of reserving an item for the cost
of a fixed fraction $\alpha$ of its size. An algorithm may pay this fraction in order to
postpone its decision on whether to include or reject the item until after the
last item of the instance was presented.

While the classical \emph{online simple knapsack problem} does not admit any constantly
bounded competitive ratio in the deterministic setting, we find that adding the
possibility of reservation makes the problem constantly competitive, with
varying competitive ratios depending on the value of the fraction $\alpha$. 
We give tight bounds for the whole range of reservation costs, which is
split up into four connected areas, depending on the value of $\alpha$:
Up to a value of $0.25$, the competitive ratio is $2$. This is followed by a
competitive ratio of $(1+\sqrt{5-4\alpha})/(2(1-\alpha))$ for values up to $\sqrt{2}-1$.
Afterwards, the competitive ratio is $2+\alpha$ for values up to $\phi -1$,
where $\phi$ is the golden ratio. Finally, the competitive
ratio diverges with a value of $1/(1-\alpha)$ for values between $\phi-1$ and $1$.

With our analysis, we find a counterintuitive characteristic of the problem:
Intuitively, one would expect that the possibility of rejecting items becomes more 
and more helpful for an online algorithm with growing reservation costs.
However, for reservation costs above $\sqrt{2}-1$, an algorithm that
is unable to reject any items tightly matches the lower bound and is thus the
best possible. On the other hand, for any positive reservation cost smaller than 
$\sqrt{2}-1$, any algorithm that is unable to reject any items
performs considerably worse than one that is able to reject.
\end{abstract}


\end{frontmatter}

\section{Introduction}
Online algorithms can be characterized by receiving their input in an
iterative fashion and having to act in an irrevocable way on each
piece of the input, \eg, by deciding whether to include an element
into the solution set or not. This has to be done with no additional
knowledge about the contents or even the length of the rest of the
instance that is still to be revealed. The goal, as with 
offline algorithms, is to optimize some objective function. In order to
measure the performance of an online algorithm, it is
compared to that of an optimal offline algorithm on the same
instance. The worst-case ratio between the performances of these algorithms over all instances
is then called the \emph{strict competitive ratio} of an online
algorithm, as introduced by Sleator and Tarjan \cite{SleatorT85}.

One of the arguably most basic and famous problems of online computation is
called the \emph{ski rental problem} \cite{Komm16}. In this problem, someone is
going on a skiing holiday of not yet fixed duration without owning the
necessary equipment. On each day, she decides, based solely on that day's
short-term weather report, whether skiing is possible on that day or not.
On each day with suitable weather, she can either buy the equipment or rent
it for that day

for a fixed percentage of the cost of buying.
Arguably, the only interesting instances are those in which a selected
number of days are suitable for skiing, followed by the rest of the
days at which she is unable to go skiing anymore. This is
simply due to the fact that, as long as she has not decided
to buy a pair of ski yet, a day at which no skiing is possible
requires no buy-or-rent decision. Thus, the problem can be simplified
as follows: The input is a string of some markers that represent days 
suitable for skiing, and as
soon as the instance ends, skiing is no longer possible.

This notion of delaying a decision for a fixed percentage of the cost is
the model that we want to study in this work. Note that, while the
ski-rental problem in its above-mentioned form only models a single
buy-or-rent decision (buying or renting a single commodity), iterated
versions of it have been discussed in the literature, with important
applications, \eg, to energy-efficient computing
\cite{ISG07,Albers10,BDKS15}. In this paper, we investigate the power of
delaying decisions for another more involved problem, namely the online
knapsack problem.

The \emph{knapsack problem} is defined as follows. Given a set of items $I$, a size
function $w\colon I \to \R$ and a gain function $g\colon I \to \R$, find a subset $S
\subseteq I$ such that the sum of sizes of $S$ is less or equal to the
so-called \emph{knapsack capacity} (which we assume to be normalized to $1$
in this paper) and the
sum of the gains is maximized. The online variant of this problem 
reveals the items of $I$ piece by piece, with an algorithm having to immediately
decide whether to pack a revealed item or to discard it.

The knapsack problem is a classical hard optimization problem, with the decision
variant being shown to be NP-complete as one of Karp's 21 NP-complete problems
\cite{Karp72}. The offline variant of this problem was studied extensively in
the past, showing for example that it admits a fully polynomial time
approximation scheme \cite{IbarraK75}.

A variant of this problem in which the gains of all items coincide with their
respective sizes is called the \emph{simple knapsack problem}; it is also NP-complete. Both variants do not
admit a constantly bounded competitive ratio \cite{Marchetti-SpaccamelaV95}. In
this paper, we focus on the online version of the latter variant, which we
simply refer to as the \emph{online knapsack problem}, for short \OKP, if not explicitly 
stated otherwise.

We propose a rather natural generalization of the online
knapsack problem.  Classically, whenever an item of the instance is
presented, an online algorithm has to irrevocably decide whether to include
it into its solution (\ie, pack it into the knapsack) or to discard it. In
our model, the algorithm is given a third option, which is to reserve
an item for the cost of a fixed percentage $0 \leq \alpha \leq 1$ of
its value.  The algorithm may reserve an arbitrary number of items of
the instance and decide to pack or discard them at any later point.
Philosophically speaking, the algorithm may pay a prize in order to
(partially) ``go offline.'' It is easy to see that, for $\alpha = 0$, the
complete instance can be learned before making a decision without any
disadvantages, essentially making the problem an offline problem,
while, for $\alpha = 1$, reserving an item is not better than discarding
it because packing a reserved item does not add anything to the gain.

This extension to the knapsack problem is arguably quite natural: Consider
somebody trying to book a flight with several intermediate stops. Since
flight prices are subject to quick price changes, it might be necessary to
invest some reservation costs even for some flights not ending up in the
final journey. The knapsack problem with reservations might also be seen as
a simple model for investing in a volatile stock market, where a specific
type of derivates is available (at some cost) that allows the investor to
fix the price of some stock for a limited period of time.

One of the key properties of our reservation model applied to the simple
knapsack problem is its unintuitive behavior with respect to rejecting items.
It shows sharp thresholds for the competitive ratio at seemingly
arbitrary points: For some interval of low reservation charges, the competitive
ratio is not affected by the charge at all, on the next interval, it follows a 
nonlinear function, penultimately it depends
linearly on the charge, and on the last interval, the competitive ratio grows
even faster. This behavior is further discussed in the following subsection.

The \OKP has been extensively studied under many other
variations, including buffers of constant size in which items may be stored
before deciding whether to pack them, studied by Han et
al.~\cite{HanKMY19}. Here, the authors allow for a buffer of size at least
the knapsack capacity into which the items 
that are presented may or may not be packed and from which a subset is then
packed into the actual knapsack in the last step. They extensively study the
case in which items may be irrevocably removed from the buffer. Our model can be
understood as having an optional buffer of infinite size, with the caveat that
each buffered item induces a cost. A variant without an additional buffer, but
with the option to remove items from the knapsack at a later point was studied 
by Iwama et al.~\cite{IwamaT02}. The same model with costs for each removal 
that are proportional to a fraction $f$ of each item from
the knapsack was researched
by Han et al.~\cite{HanKM14}, which is closer to our model. Their model allows
an algorithm to remove items that were already packed into the knapsack,
for a proportion of the value of these items. However, we 
do not know of any simple reduction from one model to the other, which is 
supported by the considerably different behavior of the competitive ratio 
relative to the reservation cost. For \OKP, Han et al. show that the
problem is 2-competitive for a removal cost factor $f \leq 0.5$ and becomes
$(1+f+\sqrt{f^2+2f+5})/2$-competitive for $f > 0.5$.

Other models of the online (simple) knapsack problem have been studied as well,
such as the behavior of the resource-augmentation model, where the knapsack of the online algorithm is
slightly larger than the one of the offline algorithm on the same instances,
studied by Iwama et al.~\cite{IwamaZ10}. 
When 
randomization is allowed, the \OKP becomes 2-competitive
using only a single random bit and does not increase its competitive ratio with
any additional bits of randomization \cite{BockenhauerKKR14}. A relatively young
measure of the informational complexity of a problem is that of advice
complexity, introduced by Dobrev, Kr\'alovi\v{c}, and
Pardubsk\'{a}~\cite{DobrevKP09}, revised by Hromkovi\v{c} et al.~\cite{HromkovicKK10} and refined by Böckenhauer et
al.~\cite{BockenhauerKKKM17}. In the standard model, an online algorithm is
given an advice tape of bits that an oracle may use to convey some bits of
information to an online algorithm about the instance in order to improve its
performance. Not surprisingly, a single advice bit also improves the competitive
ratio to 2, but interestingly, any advice in the size of $o(\log n)$ advice bits
does not improve this ratio~\cite{BockenhauerKKR14}.

\subsection{Our Contributions}
We study the behavior of the knapsack problem with reservation costs
for a reservation factor $0 < \alpha <
1$. We analyze four subintervals for $\alpha$ separately with
borders at $1/4$, $\sqrt{2}-1$ and $\phi-1\approx 0.618$,
where $\phi$ is the golden ratio.
The bounds that we are providing, which are
illustrated in Figure \ref{fig:plot}, can be
found in Table \ref{tab:compratio}.  We prove matching upper and lower bounds
for each interval.

\begin{table}[t]
\caption{Competitive ratios proven in this work.}
\label{tab:compratio}
\centering

\begin{tabular}{r l c l l}
    \hline
    \multicolumn{2}{c}{$\alpha$}&Competitive ratio&Lower bound&Upper bound\\
    \hline
    $0\;<\;$&\hspace*{-4mm}$\alpha\;\le\;0.25$&$2$&Thm.~\ref{lb:smallalpha} &Thm.~\ref{thm:ubtiny}\\
    $0.25\;<\;$&\hspace*{-4mm}$\alpha\;\le\;\sqrt{2}-1$&$\frac{1+\sqrt{5-4\alpha}}{2(1-\alpha)}$&Thm.~\ref{lb:smallalphabetter}&Thm.~\ref{thm:ubtiny}\\
    $\sqrt{2}-1\;<\;$&\hspace*{-4mm}$\alpha\;<\;\phi-1$&$2+\alpha$&Thm.~\ref{lb:smallalphabetter} &Thm.~\ref{thm:ubsmall},\;\ref{thm:ubsmall2}\\
    $\phi-1\;\le\;$&\hspace*{-4mm}$\alpha\;<\;1$&$\frac{1}{1-\alpha}$&Thm.~\ref{lb:smallalphabetter} &Thm.~\ref{thm:ublarge}\\
    \hline
\end{tabular}
\end{table}

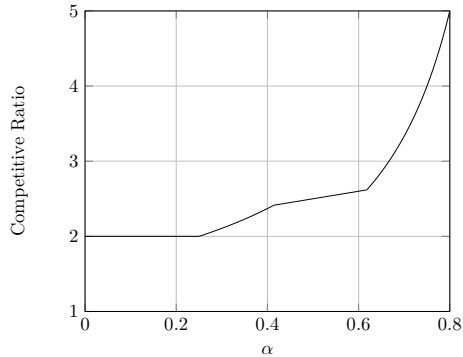
\begin{figure}
\centering
\hspace*{-6em}
  \begin{tikzpicture}[scale=0.7]
    \begin{axis}[ ymin=1,xmin=0,xmax=0.8,ymax=5, xlabel=$\alpha$, ylabel=Competitive Ratio, grid=major]
        \addplot[name path=lb1, domain=0:0.25] {2};
        \addplot[name path=lb4, domain=0.25:0.414] {(1+sqrt(5-4*x))/(2-2*x)};

        \addplot[name path=lb5, domain=sqrt(2)-1:sqrt(5)/2-1/2] {2+x};

        \addplot[name path=lb6, domain=sqrt(5)/2-1/2:0.8] {1/(1-x)};
    \end{axis}
  \end{tikzpicture} 

\caption{A schematic plot of the relation between reservation costs and the resulting competitive ratio, which is proven in this work.}
\label{fig:plot}
\end{figure}

We also take a look at a subclass of algorithms for this problem
that never discard presented items up to a point where they 
stop processing the rest of the input, with arguably paradoxical
results:  One would expect that, for very small reservation
costs, reserving items instead of rejecting them right away is more helpful
than doing the same with large reservation costs and that rejecting becomes
more and more helpful with increasing reservation costs.  But we prove that,
while, for small values of $\alpha$, every algorithm that is unable
to reject items is strictly dominated by algorithms that are able to
reject items, for larger values of $\alpha$ this is not the
case, with our algorithms used for $\alpha\ge \sqrt{2}-1$
being nonrejecting and matching their lower bounds. 

We cannot give a definitive answer on why this behavior can be observed. Our
intuition is that, for higher reservation costs, the behavior is more similar to
the model without reservations. In this classical model, any errors that occur
by ultimately not packing an item are punished severely, so while reserving
everything is costly, it is not worse than rejecting any items.

The remainder of this paper is structured as follows: In
Section~\ref{sec:ub}, we present the upper bounds for all intervals.
We divide this section by first presenting a more general algorithm
that provides us with the wanted upper bounds for values of $\alpha$ up to
$\sqrt{2}-1$ and continue with much simpler algorithms for values of
$\alpha$ larger than $\sqrt{2}-1$.  In Section~\ref{sec:lb}, we look at matching
lower bounds, with a specialized lower bound for values of $\alpha$
smaller than $0.25$ and a uniform approach for all higher values.
Section \ref{sec:nonrej} is devoted to a discussion of algorithms that
are unable to reject items. We conclude the paper in
Section~\ref{sec:conclusion}.

\subsection{Preliminaries}

Our model can be defined as follows.
Consider a knapsack of size~1 and a \emph{reservation factor} $\alpha$,
with $0\le \alpha\le 1$. 
A \emph{request sequence} $I$ is a list of item sizes
$x_1,x_2,\ldots,x_n$ with $x_i\leq1$ for $1\le i\le n$,
which arrive sequentially.
At time step~$i$, the knapsack is filled up to some size $t_i\le 1$
and the reserved items add up to size~$r_i$.
When an item with size $x_i$ arrives, an online algorithm
may \emph{pack} this item into the knapsack if $x_i+t_i\le 1$,
it may also \emph{reject} the item,
or \emph{reserve} the item at cost $\alpha\cdot x_i$.
  
At step $n+1$, no new items arrive and the knapsack contains all items
that were taken with total size~$t_n$ and the reserved items have size~$r_n=R$. 
An algorithm can additionally pack any of the reserved items which still fit
  into the knapsack, up to some size $t\le1$. 
  The \emph{gain} of an algorithm \alg solving \RKP with a reservation factor $\alpha$ on an instance $I$ is 
  $\text{gain}_{\alg}(I)= t - \alpha\cdot R$.

The \emph{strict competitive ratio} of an algorithm \alg on an instance
$I$ is, given a solution with optimal gain $\text{gain}_{\opt}(I)$ on this instance,
$\compr_{\alg}(I)={\text{gain}_{\opt}(I)}/{\text{gain}_{\alg}(I)}$. 
The general strict competitive ratio of an algorithm \alg is taken as the worst case over all possible instances,
$\compr_{\alg}=\max_{I}\{\compr_{\alg}(I)\}$.

  The strict competitive ratio as defined above is a special case of the
  well-known \emph{competitive ratio} which relaxes the above definition by
  allowing for a constant additive term in the definition. Note that this
  generalized definition only makes sense for online problems in which the
  optimal solution has unbounded gain. Since the gain of an optimal solution
  for \RKP is bounded by the knapsack capacity, we only work with
  strict competitive ratios in this paper and will simply call it
  \emph{competitive ratio} from now on. For a thorough introduction to
  competitive analysis of online algorithms, see the textbooks by Borodin
  and El-Yaniv \cite{0097013} and by Komm \cite{Komm16}.

\section{Upper Bounds}\label{sec:ub}

We start by presenting an algorithm that gives us bounds for smaller
values of $\alpha$. While the algorithm that is presented in the
following subsection does not diverge from the lower bounds for values
of $\alpha$ bigger than $\sqrt{2}-1$, we handle the two function
segments above this threshold seperately, using much simpler
algorithms. This is done to show a curious observation that is
further discussed in section~\ref{sec:nonrej}. This observation states that an
algorithm that tightly matches the lower bound for these segments no
longer has to reject any items in order to be optimal.

\subsection{Upper Bound for $0<\alpha< \sqrt{2}-1$} 

We present Algorithm~\ref{alg:alphatinyupperb} and first look at the
competitive ratio after each of the stop statements is reached. The
interesting part of the analysis is then done in
Theorem~\ref{thm:ubtiny}, which deals with the case that none of the
stop statements are reached during the run of the algorithm.

We define
$\mu=1/(\rho(1-\alpha))$, where we will choose $\rho$ as our claimed
competitive ratio, depending on the area of values for $\alpha$ that
we want to analyze. We keep track of the reserved items in the set
$R_s$. We denote the largest packing of a set $I$ of items that is of
size at most 1 by $\ensuremath{POPT}$.

The algorithm is split up into several parts, while a reservation is
only made in line~\ref{line:reserve} and if one criterion is met:
If the gain from the current item
plus the gain of everything reserved so far is still below our wanted
gain, we reserve the item.  This especially means that the current
item and the gain from our reserve itself are both smaller than
$1/\rho$ if this condition triggers.

If we do not reserve, we know that the new item together with the gain
of our reserve exceeds our wanted gain. It could however be the case
that this sum also exceeds 1 and thus does not constitute a valid
packing.  Thus, we pack in line~\ref{line:simplyfits} if the sum is\
smaller than 1 and are done or
we continue otherwise.

We also know that $R(1-\alpha) < 1/\rho$, and thus $R <
1/(\rho(1-\alpha)) = \mu$ always holds, as we do not reserve items
that would cause a violation of this condition. This, together with
the knowledge that $x_k$ is too big to pack it together with $R$
lets us derive that $x_k > 1-\mu$. Thus, if all other items from $R_s$
are smaller than $1-\mu$, we can simply remove them iteratively from
our knapsack until $x_k + R \leq 1$ holds. This leads to our wanted
result in line~\ref{line:allitemssmall}, as we know that the
reservation costs up to this point are at
most $\mu\alpha$ (since $R < \mu$), our final packing is of size at
least $\mu$ and thus our gain is $\mu - \mu\alpha = 1/\rho$.

This leaves us with a corner case in which the previous conditions
hold, except that there is at least one other item in $R_s$ which is
bigger than $1-\mu$. We then simply calculate in
line~\ref{line:tryallpackings} if we are able to
achieve our wanted competitive ratio using the new item $x_k$ anyway
and reject the item in line~\ref{line:reject} otherwise.

Thus, we are only left to analyze the competitive ratio of the
algorithm in the case that none of the \emph{stop} statements
triggered at the end of the instance, which we do in the
following theorem.

\begin{algorithm}[tb]
  \begin{algorithmic}[1]
  \State{$R:=0$;}
  \State{$R_s:=\emptyset$;}
  \For{$k=1,\ldots,n$}
    \If{$x_k + R(1-\alpha) < 1/\rho$} {$R:=R+x_k$, $R_s := R_s \cup \{x_k\}$}\label{line:reserve}
      \ElsIf{$x_k+R \leq 1$} {pack optimally and stop}\label{line:simplyfits}
      \ElsIf{$\forall j \in R_s:x_j \leq 1-\mu$} {pack optimally and stop}\label{line:allitemssmall}
      \ElsIf{$\ensuremath{POPT}(x_k \cup R_s)-R\alpha \geq 1/\rho$} {pack optimally and stop}\label{line:tryallpackings}
      \Else{ reject $x_k$}\label{line:reject}
    \EndIf
  \EndFor
  \State{pack the reserved items optimally}\label{line:packatend}
  \end{algorithmic}
  \caption{Competitive ratio of $2$ for $\alpha \leq 0.25$ and of $\frac{1+\sqrt{5-4\alpha}}{2(1-\alpha)}$ for $\alpha \leq \sqrt{2}-1$}
  \label{alg:alphatinyupperb}
\end{algorithm}

\begin{theorem}\label{thm:ubtiny}
    There is an algorithm for \RKP with a competitive ratio of $\max
    \{\frac{1+\sqrt{5-4\alpha}}{2(1-\alpha)}, 2\}$ if $0<\alpha< \sqrt{2}-1$.
\end{theorem}
\begin{proof}
    In order to show this property, we analyze
    Algorithm~\ref{alg:alphatinyupperb} and show that it provides us
    with the wanted upper bound.  We already saw in the description of the
    algorithm that in all of the \emph{stop} statements, the claim
    of the theorem holds if we set $\rho$ to our wanted
    competitive ratios. We are thus left with analyzing the
    performance of the algorithm in the case that none of the
    \emph{stop} statements are reached during its run.

    We first argue that, whenever a new item $x_k$ is rejected, the
    reserve contains exactly one item $x_i$ of size at least $1-\mu$.
    We already know that for $x_k$ to be rejected, the \emph{stop} condition in
    line~\ref{line:allitemssmall} did not trigger, which tells us that, in
    line~\ref{line:reject}, there is at least one item bigger than $1-\mu$
    in our reserve already.
    We now look at the time where the first item of size at least $1-\mu$ has
    already been reserved and assume that $x_k$ is of size at least $1-\mu$ as well.
    Then we do not reserve this item since
    \begin{align*}
        x_k + R(1-\alpha) &\geq (1-\mu) + (1-\mu)(1-\alpha) \\
                             &= \biggl( 1-\frac{1}{\rho(1-\alpha)}\biggr) (2-\alpha)\\
                             &\geq 1/\rho
    \end{align*}
    where the last inequality holds for $\rho = 2$ and all $\alpha \leq
    1-1/\sqrt{2} \approx 0.293$, in particular for $\alpha \leq 0.25$. The
    inequality remains true for $\rho = \frac{1 +
    \sqrt{5-4\alpha}}{2(1-\alpha)}$ and any $0 < \alpha < 1$, thus in particular
    for the interval considered in the theorem statement.

    Next, we argue that any item that is rejected is of size larger than $1/2$.
    This can be easily verified using the fact that the \emph{stop} statement in
    line~\ref{line:simplyfits} did not trigger, and thus $x_k > 1-R$ holds.
    Since we additionally know that $R$ contains at least one item of size at
    least $1-\mu$, it has to hold that $x_k > \mu$, which is bigger than $1/2$
    for $0 < \alpha < \sqrt{2} - 1$.
    
    This also means that any optimal solution
    can contain at most one of the items that the algorithm discards.
    Thus, we can ignore any items of an instance that both our
    algorithm and an optimal algorithm discard and fix one item of
    size larger than $1/2$ (there can only be one in any
    valid solution) as the unique item that is part of the optimal
    solution, but discarded by our algorithm.

    Let now $x_i$ be the unique item that the algorithm rejects, but that is part
    of the optimal solution and let $x_j$ be the largest reserved item of the
    algorithm at the time that $x_i$ is rejected.  As we discard
    $x_i$, we know by the conditionals that we passed and that were not
    triggered that both $x_i \leq \ensuremath{POPT}(x_k \cup R_s) \leq 1/\rho + R\alpha$ and
    $x_i + x_j > 1$ have to hold.
    We also know that the only rejected item from the optimal solution is $x_i$
    and that every other item of this solution was reserved by our algorithm.

    Let $b$ be the total size of all items reserved by our algorithm before $x_j$ arrived. On the other
    hand, let $a$ be the total size of all items reserved by our algorithm after
    item $x_i$ is rejected. At the moment that $x_i$ is rejected, we know by line~\ref{line:tryallpackings} that
    $x_i \leq \ensuremath{POPT}(x_k \cup R_s) \leq 1/\rho + (x_j + b)\alpha$. We use this term
    between steps (\ref{eq:firststep}) and (\ref{eq:secondstep}) of the following calculation. Together with $x_i + x_j
    > 1$, we can derive $1 - x_j < x_i \leq 1/\rho + (x_j + b)\alpha$ which
    solved for $x_j$ yields $\frac{1-1/\rho-b\alpha}{1+\alpha}\leq x_j$. This is
    used between steps (\ref{eq:seventhstep}) and (\ref{eq:eigthstep}) of the following inequation chain. We can now derive the
    following. To summarize, the optimal algorithm has a gain of $x_i + b + a$.
    Algorithm~\ref{alg:alphatinyupperb} also has a gain of
    $a+b$, but packs $x_j$ instead of $x_i$ and has to pay reservation
    costs of $\alpha$ for $a$, $b$ and $x_i$.  Hence, we have 

    \begin{align}
        \frac{gain_\opt}{gain_\alg} &\leq \frac{x_i + b + a}{(1-\alpha)(x_j + b + a)}& \text{Use: $x_i \leq 1/\rho + (x_j+b)\alpha$}\label{eq:firststep}\\
                          &\leq \frac{1}{1-\alpha}\cdot\frac{1/\rho+(x_j + b)\alpha + b + a}{x_j + b + a}& \text{Add $a\cdot\alpha$ to numerator}\label{eq:secondstep}\\
                          &\leq \frac{1}{1-\alpha}\cdot\frac{1/\rho+x_j\alpha + (b + a)(1+\alpha)}{x_j + b + a}& \text{Take $1+\alpha$ as common factor}\label{eq:thirdstep}\\
                          &= \frac{1+\alpha}{1-\alpha}\cdot\frac{\frac{1}{\rho(1+\alpha)}+x_j \frac{\alpha}{1+\alpha} + b + a}{x_j + b + a}.& \label{eq:fourthstep}
    \end{align}

    If the fraction to the right of $\frac{1+\alpha}{1-\alpha}$ is at most 1, we are done: It holds
    that $\frac{1+\alpha}{1-\alpha} \leq 2$ for $\alpha \leq 1/3$ and
    also $\frac{1+\alpha}{1-\alpha} \leq
    \frac{1+\sqrt{5-4\alpha}}{2(1-\alpha)}$ for $\alpha \leq
    \sqrt{2}-1$. Otherwise, we use that removing the term $a+b$ from
    both the numerator and the denominator can only increase the value
    of the right fraction and continue as follows.

    \begin{align}
        \frac{gain_\opt}{gain_\alg} &\leq \frac{1+\alpha}{1-\alpha}\cdot\frac{\frac{1}{\rho(1+\alpha)}+x_j\frac{\alpha}{1+\alpha}}{x_j}& \text{Simplify equation}\\
                          &= \frac{\frac{1}{\rho x_j}+\alpha}{1-\alpha}& \text{Use $x_j \geq \frac{1-1/\rho+b\alpha}{1+\alpha}$}\label{eq:seventhstep}\\
                          &\leq \frac{\frac{1+\alpha}{\rho -1 +\rho b \alpha}+\alpha}{1-\alpha}& \text{Remove $\rho b\alpha$}\label{eq:eigthstep}\\
                          &\leq \frac{\frac{1+\alpha}{\rho -1}+\alpha}{1-\alpha}& \text{Replace $\rho$ by 2 ($\rho \geq 2$)}\\
                          &\leq\frac{1+2\alpha}{1- \alpha}& 
    \end{align}
    This gives us our wanted competitive ratio of at most $\rho=2$ for
    $\alpha \le 0.25$ and likewise of 
    $\rho=\frac{1+\sqrt{5-4\alpha}}{2(1-\alpha)}$ for $0.25 \leq \alpha
    \leq \sqrt{2}-1$.
\end{proof}
\subsection{Upper Bound for $\frac16<\alpha< \phi -1$}

We now prove a general upper bound for $\alpha<\phi -1$. We split the
proof, into two pieces: Theorem~\ref{thm:ubsmall} handles the case for
values of $\alpha$ up to $\frac12$, . Theorem~\ref{thm:ubsmall2}
contains an induction over the number of large elements in an
instance, which proves the upper bound for the rest of the interval up
to $\phi-1$.

\begin{theorem}\label{thm:ubsmall}
There exists an algorithm for \RKP with a  competitive ratio of at most $2+\alpha$ if $0<\alpha\le\frac12$.
\end{theorem}

\begin{theorem}\label{thm:ubsmall2}
There exists an algorithm for \RKP with a  competitive ratio of at most $2+\alpha$ if $\frac12<\alpha <\phi-1$.
\end{theorem}

\begin{algorithm}[tb]
\begin{algorithmic}
\State{$R:=0$;}
\For{$k=1,\ldots,n$}
  \If{$x_k + (1-\alpha) R \geq 1/(2+\alpha)$}
    \State{pack $x_1,\ldots,x_k$ optimally;}
    \State{stop}
  \Else
    \State{reserve $x_k$;}
    \State{$R:=R+x_k$}
  \EndIf
\EndFor
\State{pack $x_1,\ldots,x_n$ optimally}
\end{algorithmic}
\caption{Competitive ratio of $2+\alpha$ for
$0 < \alpha \leq \phi -1$}
\label{fig:Asub}
\end{algorithm}

We consider Algorithm~\ref{fig:Asub}, which, unlike
Algorithm~\ref{alg:alphatinyupperb}, does not reject any offered item 
until it reaches the desired gain and stops processing the rest of the input. In
Section~\ref{sec:nonrej}, we further discuss this class of algorithms 
and when they are optimal.
We need the following technical lemmas.

\begin{lemma}\label{lem:res_size}
  The size of $R$ in Algorithm~\ref{fig:Asub} is never larger than
  $1/{(2+\alpha)(1-\alpha)}$. 
\end{lemma}
\begin{proof}
  Let us consider a request sequence $x_1,\hdots, x_n$ for
  Algorithm~\ref{fig:Asub}, such that $x_k$ triggers the packing, \ie, $x_k +
  (1-\alpha) R \geq 1/(2+\alpha)$.
    
  We know that 
  $x_{k-1}+(1-\alpha)R'<\frac{1}{2+\alpha},$
  where $R$ is defined as $R=x_{k-1}+R'$. But $x_{k-1}\ge (1-\alpha)x_{k-1}$ which means that
  $(1-\alpha)R=(1-\alpha)x_{k-1}+(1-\alpha)R'\le x_{k-1}+(1-\alpha)R'<\frac{1}{2+\alpha}.$
  Thus,
  $R<\frac{1}{(2+\alpha)(1-\alpha)}.$
    
  Observe that this is, by construction, the largest possible value that $R$
  can take before the algorithm triggers the packing, thus, we have proved
  the claim.
\end{proof}

Note that, for every considered value of $\alpha$, the upper bound on $R$ is
positive, that is, $(2+\alpha)(1-\alpha)> 0$ if $\alpha$ is smaller than
1.  With Lemma~\ref{lem:res_size}, we can prove the following claim.

\begin{lemma}\label{lem:minpack}
  If Algorithm~\ref{fig:Asub} packs at least $1/{(2+\alpha)(1-\alpha)}$
  into the knapsack, then its competitive ratio is at
  most $2+\alpha$.
\end{lemma}
\begin{proof}
  We know that, if the algorithm packs at least
  $\frac{1}{(2+\alpha)(1-\alpha)}$, the gain of the algorithm, $\rm{gain}_{\alg}$, is
  at least $\frac{1}{(2+\alpha)(1-\alpha)}-\alpha R.$
  But, by Lemma~\ref{lem:res_size}, $R<1/(2+\alpha)(1-\alpha)$, which means
  that
  $\rm{gain}_{\alg}(I)\ge \frac{1}{(2+\alpha)(1-\alpha)}-
  \frac{\alpha}{(2+\alpha)(1-\alpha)}=\frac{1}{2+\alpha},$
  providing us with the desired upper bound on the competitive ratio.
\end{proof}

The next lemma allows us to restrict our attention to ordered sequences of
items. It proves that, if an instance violates
the upper bound for Algorithm~\ref{fig:Asub}, we can
rearrange the first $k-1$ items in decreasing order and get another
hard instance.

\begin{lemma}\label{lem:order}
  If $x_1,\ldots,x_n$ is an instance for Algorithm~\ref{fig:Asub}, whose
  packing is triggered by $x_k$, with a competitive ratio larger than
  $2+\alpha$, then there is another instance
  $x_{i_1},x_{i_2},x_{i_3},\ldots,x_{i_{k-1}},\allowbreak x_k, \hdots, x_n$ where
  $(i_1,i_2,\ldots,i_{k-1})$ is a permutation of $(1,\ldots,k-1)$ and
  $x_{i_1}\geq x_{i_2}\geq\ldots\geq x_{i_{k-1}}$, that also has a
  competitive ratio larger than $2+\alpha$.
\end{lemma}
\begin{proof}
  We show how to transform the original sequence such that the number of
  inversions is reduced by one. Repeating this transformation yields the
  desired result.
  
  Let $x_1,\ldots,x_n$ be an instance for Algorithm~\ref{fig:Asub}, such that
  packing is triggered by $x_k$ with a competitive ratio larger than
  $2+\alpha$. Assume that $x_i < x_{i+1}$ for some $i<k-1$. We prove that we
  will get again an instance with a competitive ratio larger than $2+\alpha$
  if we swap the positions of $x_i$ and $x_{i+1}$. It is easy to see that
  this transformation preserves the sum of the first $k-1$ items. To show
  that it is a sequence that does not trigger the packing after the $i$th or $(i+1)$st item,
  we have to verify that
  \begin{equation}\label{eq:shorter-sum}
  x_{i+1}<\frac1{2+\alpha}-(1-\alpha)(x_1+\cdots+x_{i-1})
  \end{equation}
  and
  \begin{equation}\label{equ:swap}
  x_i<\frac1{2+\alpha}-(1-\alpha)(x_1+\cdots+x_{i-1}+x_{i+1})\;.
  \end{equation}
  
  Here, (\ref{eq:shorter-sum}) is straightforward, since $x_k$ for $k>i+1$ triggering the
  packing in the unmodified instance implies
  $x_{i+1}<\frac1{2+\alpha}-(1-\alpha)(x_1+\cdots+x_i)$. We prove
  (\ref{equ:swap}) by means of a contradiction. Assume, that
  \begin{equation}\label{eq:swap-contradiction}
  x_i\geq\frac1{2+\alpha}-(1-\alpha)(x_1+\cdots+x_{i-1}+x_{i+1})\;.
  \end{equation}
  
  We also know that, for $1\le j\le i+1$, 
  
  \begin{equation}\label{eq:prev-req}
    x_j\leq \frac1{2+\alpha}-(1-\alpha)(x_1-\hdots-x_{j-1})\;,
  \end{equation}
  
  and by construction 
  
  \begin{equation}\label{eq:cond}
    x_i < x_{i+1}\;.
  \end{equation}
  
  If we define $z=\frac1{2+\alpha}-(1-\alpha)(x_1+\cdots+x_{i-1})$, then we can
  rewrite (\ref{eq:swap-contradiction}), (\ref{eq:prev-req}) for $j=i+1$, and
  (\ref{eq:cond}) as
  \begin{align*}
  % x_i,x_{i+1}&\geq0\\
  x_i&\geq z - \textstyle(1-\alpha) x_{i+1}\\
  % x_i&\leq z\\
  x_{i+1}&\leq z-\textstyle(1-\alpha) x_i\\
  x_i&< x_{i+1}\;.
  \end{align*}
  If we modify the first two equations to isolate the term
  $(1-\alpha)^{-1}z$, we obtain
  \begin{align*}
  \frac{x_i}{(1-\alpha)}+x_{i+1}&\geq \frac{z}{(1-\alpha)}\\
  \frac{x_{i+1}}{(1-\alpha)}+x_i&\leq \frac{z}{(1-\alpha)}\;,
  \end{align*}
  which means that
  \begin{equation*}
    \frac{x_{i+1}}{(1-\alpha)}+x_i\leq \frac{x_i}{(1-\alpha)}+x_{i+1}\;,
  \end{equation*}
  that is, these two equations are satisfiable if and only if $x_{i+1}\leq x_i$.
  
  It remains to show that the new sequence has a competitive ratio larger
  than $2+\alpha$. Although the order of $x_1,\ldots,x_{k-1}$ has been
  changed, the requests are still the same and all of them are reserved and
  trigger a packing in $x_k$ with the same reservation size, just as before
  the transformation. Hence, if there was no way to pack the items well
  enough before the transformation, it also cannot be done afterwards.
\end{proof}

From Lemma~\ref{lem:minpack}, we know that, if Algorithm~\ref{fig:Asub}
packs at least $1/(2+\alpha)(1-\alpha)$, this guarantees a
competitive ratio of at most $2+\alpha$.  Thus, if we have enough elements
that are smaller than the gap of size $1-1/(2+\alpha)(1-\alpha)$, those elements
can be packed greedily and always achieve the desired competitive ratio.
Let us call \emph{small items} those of size smaller than\looseness-1
\begin{equation}
1-\frac{1}{(2+\alpha)(1-\alpha)}=
\frac{(2+\alpha)(1-\alpha)-1}{(2+\alpha)(1-\alpha)}=
\frac{1-\alpha-\alpha^2}{(2+\alpha)(1-\alpha)}\;
\label{equ:smallitems}
\end{equation}

This definition is only valid when (\ref{equ:smallitems}) is positive, that is,
when $1-\alpha-\alpha^2\ge0$, which is the case if $0<
\alpha\le \phi-1$, including our desired range. We
call \emph{large items} those of larger size.
Let \fourfourfive be the unique positive real root of the polynomial
$1-2\alpha-\alpha^2+\alpha^3$, \ie,
\[
\fourfourfive = \frac13 +
\frac{2\sqrt{7}}{3}\cos{\biggl(\frac13\arccos{\biggl(-\frac{1}{2\sqrt{7}}\biggr)}-
    \frac{2\pi}{3}\biggr)}\approx 0.445\;.
\]

\begin{lemma}\label{lem:largeitems}
  Given any request sequence $x_1\ge x_2\ge \hdots\ge x_{k-1},
  x_k,\hdots,x_n$, where $x_k$ triggers the
  packing in Algorithm~\ref{fig:Asub},
  %write more general!!
  there is at most one large item if $0<\alpha\le\fourfourfive$ and there
  are at most two large items if $\fourfourfive< \alpha\le 0.5$.
\end{lemma}
\begin{proof}
Assume by contradiction that there exists a request
sequence where $x_1\ge x_2\ge \hdots \ge x_i \ge
\frac{1-\alpha-\alpha^2}{(2+\alpha)(1-\alpha)}$, with $k>i$. Any item $x_j$
with $j\le i$ satisfies
  \[x_{j}\le \frac1{2+\alpha}- (1-\alpha)(x_1+x_2+\hdots+x_{j-1})\;,\]
  in particular,
  \begin{equation}\label{eq:xi}
  x_{i}\le \frac1{2+\alpha}- (1-\alpha)(x_1+x_2+\hdots+x_{i-1})\;.
  \end{equation}
  
  All of the contributions of previous requests are negative, so in order to
  obtain a maximal value for $x_i$, we need that $x_1,\hdots, x_{i-1}$ are
  minimal, but by construction, still greater or equal than $x_i$. Thus, the
  maximal value is obtained when $x_1=x_2=\hdots=x_i$. In this case, we
  obtain from (\ref{eq:xi}) the upper bound
  $x_{i}\le \frac1{2+\alpha}- (1-\alpha)(i-1)x_i$
  which we solve for $x_i$ and obtain
  $x_{i}\le \frac{1}{(2+\alpha)(i(1-\alpha)+\alpha)}$.

  We also know that $x_i$ is a large item, thus we can also state 
  $x_i \ge \frac{1-\alpha-\alpha^2}{(2+\alpha)(1-\alpha)}$ as a lower bound for $x_i$.

  Thus we get
  $\frac{1}{(2+\alpha)(i(1-\alpha)+\alpha)}\ge
  \frac{1-\alpha-\alpha^2}{(2+\alpha)(1-\alpha)}$
  which we solve for $i$ to obtain
  $i\le %\frac{1-2\alpha+\alpha^2+\alpha^3}{(1-\alpha)(1-\alpha-\alpha^2)}=
    1+\frac{\alpha^2}{(1-\alpha)(1-\alpha-\alpha^2)}$.
  In particular, for $i=2$ we get
    $2(1-\alpha)(1-\alpha-\alpha^2)\le(1-\alpha)(1-\alpha-\alpha^2)+\alpha^2$
    which is equivalent to
    $(1-\alpha)(1-\alpha-\alpha^2)\le\alpha^2$
    and thus to
%   1-2\alpha+\alpha^3&\le\alpha^2\\
    $1-2\alpha-\alpha^2+\alpha^3\le0$,
which means that the number of large items is strictly smaller than $2$
  for $\alpha\le\fourfourfive$.

  For $i=3$, we get
    $3(1-\alpha)(1-\alpha-\alpha^2)\le(1-\alpha)(1-\alpha-\alpha^2)+\alpha^2$
    or equivalently
    $2(1-2\alpha+\alpha^3)\le\alpha^2$.
    Hence,
    $2-4\alpha-\alpha^2+2\alpha^3\le 0$, 
    which means that the number of large items is strictly smaller than 3 for
  $\alpha \le 0.5$, since $0.5$ is the unique positive real root of the
left-hand-side polynomial.
\end{proof}

Now we are ready to prove the claimed competitive ratio of Algorithm~\ref{fig:Asub}.

\begin{proof}[Proof of Theorem~\ref{thm:ubsmall}]
 Let us consider first the case where $\alpha\le\fourfourfive$, which means
  that only one large element can appear in the request sequence without
  triggering the packing. We know that,
  if $x_1\ge 1/(2+\alpha)$, then the algorithm will take it and our
  competitive ratio will be at most $2+\alpha$, and if a request sequence
  ends without triggering the packing, then the optimum would be
  $x_1+\hdots+x_n=s<1$ and Algorithm~\ref{fig:Asub} would have a gain of
  $(1-\alpha)s$, so the competitive ratio will be at most
  $1/(1-\alpha)<2+\alpha$.
  
  Otherwise, assume that there is a shortest sequence
  containing at least two elements, which at some point, namely with request
  $x_k$, triggers Algorithm~\ref{fig:Asub} to pack. We now consider all
  possible cases.
  
  If $x_1+\hdots+x_k\le 1$, then we know that the gain of the algorithm is
  $x_k+(1-\alpha)(x_1+\hdots+x_{k-1})$, which is larger than $1/(2+\alpha)$
  by construction.
  
  We are now in the case where $x_1+\hdots+x_k>1$. If all of the elements of
  the request sequence are small, Algorithm~\ref{fig:Asub} packs $x_k$ and
  then it can greedily pack elements from $x_1,\hdots, x_{k-1}$. Because these
  elements all have size smaller than $1-\frac{1}{(2+\alpha)(1-\alpha)}$, we
  know that Algorithm~\ref{fig:Asub} will be able to pack at least
  $\frac{1}{(2+\alpha)(1-\alpha)}$ into the knapsack, obtaining a competitive
  ratio of at most $2+\alpha$ by Lemma~\ref{lem:minpack}.
  
  If not all elements are small, we know that at most $x_1$ is a large
  element, and we do a further case distinction. If $x_1+x_k<1$, then we pack
  them, and we greedily pack the rest of the items, obtaining a competitive
  ratio of at most $2+\alpha$ with the same argument as in the previous case.
  If $x_1+x_k>1$, then we consider two cases. If
  $x_2+\hdots+x_k>\frac{1}{(2+\alpha)(1-\alpha)}$, we can again greedily pack
  the small items until we obtain the desired competitive ratio. Otherwise,
  we know that $x_k>1-x_1$ and
  $x_2+\hdots+x_k<\frac{1}{(2+\alpha)(1-\alpha)}$ and thus $x_2,\ldots,x_k$
  can all be packed into the knapsack. In this case, we can calculate the
  gain of Algorithm~\ref{fig:Asub} as
  \begin{align*}
    & \hspace*{-1em}x_2+\hdots+x_k -\alpha (x_1+\hdots+x_{k-1})\\
    &= x_k + (1-\alpha)(x_2+\hdots+x_{k-1}) -\alpha x_1\\
    &> 1 - x_1 -\alpha x_1 +(1-\alpha)(x_2+\hdots+x_{k-1})\\
    &= 1- (1+\alpha) x_1 +(1-\alpha)(x_2+\hdots+x_{k-1})\\
    &> 1- \frac{(1+\alpha)}{(2+\alpha)} +(1-\alpha)(x_2+\hdots+x_{k-1})\\
    &= \frac{1}{(2+\alpha)}+(1-\alpha)(x_2+\hdots+x_{k-1})\;,
  \end{align*}
  where we used the fact that $x_k>1-x_1$, and that $x_1<1/(2+\alpha)$ since
  it did not trigger the packing. This gain is larger than $1/(2+\alpha)$ and
  we obtain the desired bound for the competitive ratio.
  
  Now, let us consider the case where $\alpha \le 0.5$, which means that the
  request sequence either has at most two large elements or it is shorter
  than 3 elements. If we get a request sequence of $2$ elements, the first
  one has to be $x_1<1/(2+\alpha)$, otherwise we already have the desired
  competitive ratio by taking that element. Now, if $x_2\le 1-1/(2+\alpha)$,
  then both elements fit into the knapsack and we get again a competitive
  ratio of at most $1/(1-\alpha)<2+\alpha$. Otherwise, $x_2>1-1/(2+\alpha)$,
  so the algorithm can pack $x_2$ and obtain a gain of at least
  $1-1/(2+\alpha)-\alpha/(2+\alpha)=1/(2+\alpha)$.
  
  If the request sequence has at least length 3, we do a further case
  distinction. As in the previous case, if a request sequence ends without
  triggering the packing, then the optimal packing is $x_1+\hdots+x_n=s<1$
  and Algorithm~\ref{fig:Asub} has a gain of $(1-\alpha)s$, so the
  competitive ratio is at most $1/(1-\alpha)<2+\alpha$. Otherwise, assume
  that there is a shortest contradictory sequence which at some point $x_k$
  triggers Algorithm~\ref{fig:Asub} to pack. If $x_1+\hdots+x_k\le 1$, then we
  know that the gain of the algorithm is
  $x_k+(1-\alpha)(x_1+\hdots+x_{k-1})$, which is larger than $1/(2+\alpha)$
  by construction.
  
  Otherwise, $x_1+\hdots+x_k>1$. In this case, if all of the elements of the
  request sequence are small or there is at most one large element, we are in
  the same case as before
  %HJB: Maybe enumerating the cases would be helpful?
    and we obtain the desired competitive ratio. Let us
  now assume that we have two large elements in our request sequence. We know
  that  $x_1<\frac{1}{2+\alpha}$
  and $x_2<1/(2+\alpha)-(1-\alpha)x_1$. By Lemma~\ref{lem:order}, 
  we can also assume that $x_1\ge x_2$ and with this we can
  deduce that $x_2< \frac{1}{(2+\alpha)}-(1-\alpha)x_2\nonumber$ and thus
  $(2-\alpha)x_2<\frac{1}{(2+\alpha)}$
  which allows us to estimate $x_2$ as $x_2<\frac{1}{(2+\alpha)(2-\alpha)}$
  Now, we can fall in one of the following cases: 
    \begin{enumerate}
    \item If $x_1+x_2+x_k\le 1$, filling up greedily from the small items, we
      obtain the desired competitive ratio.
    \item If $x_1+ x_2+ x_k>1$, but $x_1+x_k\le1$, we pack $x_1$ and $x_k$ and
      pack the small items greedily. If the small items fill the knapsack until at
      least $\frac{1}{(2+\alpha)(1-\alpha)}$, we know that we obtain the
      desired competitive ratio. Otherwise, the gain obtained by the algorithm
      is
    \begin{align*}
    & \hspace*{-1em}x_1+x_3+\hdots+x_k -\alpha (x_1+\hdots+x_{k-1})\\
    &= x_k + (1-\alpha)(x_1+x_3+\hdots+x_{k-1}) -\alpha x_2\\
    &> 1 - x_1 - x_2 +(1-\alpha) x_1 -\alpha x_2 +(1-\alpha)(x_3+\hdots+x_{k-1})\\
    &= 1- \alpha x_1 - (1+\alpha) x_2 +(1-\alpha)(x_3+\hdots+x_{k-1})\\
    &> 1- \frac{\alpha}{2+\alpha}-\frac{1+\alpha}{(2+\alpha)(2-\alpha)} +(1-\alpha)(x_3+\hdots+x_{k-1})\\
    &= \frac{(2+\alpha)(2-\alpha)-\alpha(2-\alpha)-(1+\alpha)}{(2+\alpha)(2-\alpha)}+(1-\alpha)(x_3+\hdots+x_{k-1})\\
    &= \frac{3(1-\alpha)}{(2+\alpha)(2-\alpha)}+(1-\alpha)(x_3+\hdots+x_{k-1})\\
    &\ge \frac{1}{2+\alpha}+(1-\alpha)(x_3+\hdots+x_{k-1})\;.
    \end{align*}
    where the last step is true for any $\alpha \le 1/2$, and we used
    $x_k>1-x_1-x_2$ and the upper bounds $x_1<\frac{1}{2+\alpha}$ and
    $x_2<\frac{1}{(2+\alpha)(2-\alpha)}$ for $x_1$ and $x_2$.
  \item If $x_1 + x_k>1$, but $x_2+x_k\le1$, we pack $x_2$ and $x_k$ into the
    knapsack and the small items greedily. Here, again, if the small items
    fill the knapsack until at least $\frac{1}{(2+\alpha)(1-\alpha)}$, we
    know that we obtain the desired competitive ratio, otherwise, the gain
    obtained by the algorithm is $ x_2\hdots+x_k -\alpha (x_1+\hdots+x_{k-1})$
    which is exactly the same as in the case for $\alpha\le\fourfourfive$,
  %HJB: Which is the exact case that is referred to? Cite some equation number.
    and we can perform the same operations and obtain the desired bound for
    the competitive ratio.
  \item If $x_2+x_k>1$, we pack $x_k$ into the knapsack and pack the small items
    greedily. If these fill the knapsack with at least
    $\frac{1}{(2+\alpha)(1-\alpha)}$, we know that we obtain the desired
    competitive ratio, otherwise, the gain obtained by the algorithm is
    \begin{align*}
    & \hspace*{-1em}x_3+\hdots+x_k -\alpha (x_1+\hdots+x_{k-1})&\\
    &= x_k + (1-\alpha)(x_3+\hdots+x_{k-1}) -\alpha x_2-\alpha x_1& \text{$x_k > 1-x_2$, remove $(1-\alpha)x_1$}\\
    &\geq 1 - x_2 -\alpha x_1 -\alpha x_2 +(1-\alpha)(x_3+\hdots+x_{k-1})&\\
    &= 1- \alpha x_1 - (1+\alpha) x_2 +(1-\alpha)(x_3+\hdots+x_{k-1})\;,&
    \end{align*}
    which is exactly the same as in the case where $x_1+ x_2+ x_k>1$ but
    $x_1+x_k\le1$.
  \end{enumerate} 
\end{proof}

We continue with proving an upper bound for the rest of the interval, using 
induction over the number of large elements. Before we start our proof, we provide a
lemma that allows us to ignore possible small elements during the proof of
Theorem~\ref{thm:ubsmall2}.

\begin{lemma}\label{lem:smallelements} 
Given a request sequence without small
elements for which Algorithm~\ref{fig:Asub} does not achieve a competitive
ratio of $2+\alpha$, adding small elements to it will only improve its
competitive ratio.
\end{lemma}
\begin{proof}
Let us consider a request sequence containing only large elements $x_1\ge \hdots
\ge x_{k-1},x_k,\hdots,x_n$, where $x_k$ is the element triggering the packing
for  Algorithm~\ref{fig:Asub}, and for which the achieved competitive ratio is larger than
$2+\alpha$. This means, by Lemma~\ref{lem:minpack}, that the total size of the items packed into the knapsack is
smaller than $1/(2+\alpha)(1-\alpha)$. By definition, if we add enough small
elements to the request sequence before $x_k$, the small elements can be packed
greedily until the knapsack is filled up to $1/(2+\alpha)(1-\alpha)$, achieving the
desired competitive ratio. If not enough of them are added before $x_k$, the
small elements requested before $x_k$ will be reserved but will still be able to
be packed, so they will never contribute negatively to the total packing gain.
\end{proof}

This lemma brings us to the point where it is possible to prove Theorem~\ref{thm:ubsmall2}.

\begin{proof}[Proof of Theorem~\ref{thm:ubsmall2}]
We prove by induction that, for any finite number of large elements before the packing, Algorithm~\ref{fig:Asub} achieves the desired competitive ratio for $\alpha <\phi-1$.
The base case for zero large elements is trivial, as we can greedily reserve and later pack all small elements to get the desired competitive ratio.
Let us assume that Algorithm~\ref{fig:Asub} achieves the desired competitive ratio for any request sequence with less than $k-1$ large elements before the algorithm packs and stops.

If a request sequence has $k-1$ large elements, we can assume by
Lemma~\ref{lem:order} that the smallest of those is $x_{k-1}$, and we
can also assume by
Lemma~\ref{lem:smallelements} that $x_{k}$ triggers the packing. 
 Let us assume, using Lemma~\ref{lem:order2}, that we have an instance 
  $x_1\ge x_2 \ge \hdots\ge x_{k-1}, x_k, \hdots x_n$. 
Moreover, note that 
\begin{align}
\frac{\alpha}{1-\alpha}  &< 1-\alpha \text{ for }\alpha < \phi -1\;.\label{eq:boundalpha}.
\end{align}
 We distinguish two cases. 
\begin{enumerate}
\item Assume that, when the packing is triggered, $x_{k-1}$ is not part of an optimal packing.

In this case, 
\begin{equation}\label{eq:condition}
  \biggl(\sum_{i\text{ is packed}}x_i\biggr)+x_{k-1}>1\;.
\end{equation} 
We also have 
\begin{align}
x_{k-1}+(1-\alpha)\sum_{j<k-1}x_j&<\frac{1}{2+\alpha}\;,\label{eq:bound1}
\end{align}
since $x_{k-1}$ would trigger the packing otherwise, and thus 
\begin{align}
x_{k-1}&< \frac1{2+\alpha}\;.\label{eq:bound2}
\end{align}
Thus, the gain that Algorithm~\ref{fig:Asub} achieves is 
\begin{flalign*}
& \biggl(\sum_{j\text{ is packed}}x_j\biggr)-\alpha R&& \\
& > 1 - x_{k-1} -\alpha \sum_{j\le k-1} x_j && \text{using (\ref{eq:condition})} \\
& = 1 - (1+\alpha) x_{k-1} - \alpha \sum_{j < k-1} x_j &&  \\
& = 1 - 2\alpha x_{k-1} -(1-\alpha)\biggl( x_{k-1}
+\frac{\alpha}{1-\alpha}\sum_{j < k-1} x_j \biggr) && \\
& \geq 1 - 2\alpha x_{k-1} -(1-\alpha)\biggl( x_{k-1}
+(1-\alpha)\sum_{j < k-1} x_j \biggr) && \text{using (\ref{eq:boundalpha})}\\
& \geq 1 - 2\alpha x_{k-1} -(1-\alpha)\frac{1}{2+\alpha} && \text{using (\ref{eq:bound1})} \\
& > 1 - \frac{2\alpha}{2+\alpha} -\frac{1-\alpha}{2+\alpha} && \text{using (\ref{eq:bound2})}\\
& = \frac{1}{2+\alpha} \;,&&
\end{flalign*}
as we wanted.

\item Assume that when the packing is triggered, $x_{k-1}$ is part of an optimal packing. 
We consider two subcases. 
\begin{enumerate}
\item Taking $x_{k-1}$ out of the request sequence still triggers the packing.

This means that 
$x_{k}+(1-\alpha)(R-x_{k-1})\ge \frac{1}{2+\alpha}$
holds. We can thus consider the sequence $x_1,\hdots,x_{k-2},x_{k}$. This sequence has $k-2$ large elements, and $x_{k-1}$ cannot be part of its optimal solution, thus its competitive ratio is at most $2+\alpha$ by the induction hypothesis, and the competitive ratio after adding $x_{k-1}$ can only get better.
\item Taking $x_{k-1}$ out of the request sequence does not trigger the packing.

This means that
\begin{align}
 x_{k}+(1-\alpha)(R-x_{k-1})< \frac{1}{2+\alpha}\;,&\text{ and also}\label{eq:condition2}\\
 x_{k}+(1-\alpha)R\ge \frac{1}{2+\alpha}\;,&\label{eq:condition3}
\end{align}

and, if we let $x_t$ be the smallest element that does not get packed, $x_t\ge x_{k-1}$ holds. Also,
because of the optimality of the packing $x_{k}\ge x_t$, (otherwise one can take all of the reserved elements as the packing and obtain a better bound) and
\begin{equation}\label{eq:condition4}
 \sum_{j\text{ is packed}}x_j-x_{k-1}+x_t>1\;
\end{equation}
holds.
With these bounds, the gain incurred by the algorithm is at least
\begin{flalign*}
& \hspace*{-3em} \biggl(\sum_{j\text{ is packed}}x_j\biggr)-\alpha R\\
& > 1- x_t +  x_{k-1} -\alpha R &&\text{using (\ref{eq:condition4})} \\
& > 1 -x_t + \frac{x_{k}}{1-\alpha} + R - \frac{1}{(1-\alpha)(2+\alpha)} -\alpha R &&\text{using (\ref{eq:condition2})} \\
& = 1 -x_t + \frac{x_{k}}{1-\alpha} + (1 - \alpha) R - \frac{1}{(1-\alpha)(2+\alpha)}&&\\
& \geq 1 -x_t + \frac{\alpha x_{k}}{1-\alpha} +\frac{1}{2+\alpha}  - \frac{1}{(1-\alpha)(2+\alpha)}&&\text{using (\ref{eq:condition3})} \\
& = \frac{1}{2+\alpha}+\frac{1-\alpha-\alpha^2}{(1-\alpha)(2+\alpha)}-x_t+\frac{\alpha x_{k}}{1-\alpha}&& \\
& \geq \frac{1}{2+\alpha}+\frac{1-\alpha-\alpha^2}{(1-\alpha)(2+\alpha)}+\frac{\alpha -(1-\alpha)}{1-\alpha}x_{k}&& \text{using }x_k\ge x_t\\
& = \frac{1}{2+\alpha}+\frac{1-\alpha-\alpha^2}{(1-\alpha)(2+\alpha)}+\frac{2\alpha -1}{1-\alpha}x_{k}&&\\
& \geq \frac{1}{2+\alpha}\;,
\end{flalign*}
where the last step is trivially true for any $\alpha\ge 1/2$. Thus we get the desired competitive ratio in the considered range for $\alpha$.
\end{enumerate}
\end{enumerate}
This proves the induction step, and thus the desired upper bound on the competitive ratio.
\end{proof}

\subsection{Upper Bound for $\phi-1\le\alpha< 1$}

For proving an upper bound, we consider Algorithm~\ref{fig:Alub} and first bound
its reservation costs.

\begin{algorithm}[tb]
\begin{algorithmic}
\State{$R:=0$;}
\For{$k=1,\ldots,n$}
  \If{$x_k + (1-\alpha) R \geq 1-\alpha$}
    \State{pack $x_1,\ldots,x_k$ optimally;}
    \State{stop}
  \Else
    \State{reserve $x_k$;}
    \State{$R:=R+x_k$}
  \EndIf
\EndFor
\State{pack $x_1,\ldots,x_n$ optimally}
\end{algorithmic}
\caption{Algorithm for $\phi-1 \leq \alpha < 1$}
\label{fig:Alub}
\end{algorithm}

\begin{lemma}\label{lem:res_size_bigalpha}
  For Algorithm~\ref{fig:Alub}, the reservation cost $R$ is never
  larger than $1$.
\end{lemma}
\begin{proof}
  For any $j$, let $R_j$ denote the reservation cost of the algorithm after
  the items $x_1,\ldots,x_j$ have been presented.  
  If, for any $k$,
  \[x_{k}+(1-\alpha)R_{k-1}<1-\alpha\;,\] then the algorithm reserves the
  item $x_{k}$ and has a new reservation cost of $R_{k} = x_{k} +
  R_{k-1}$. Since obviously $x_{k}\ge (1-\alpha)x_{k}$, we have
  \[(1-\alpha)R_{k}=(1-\alpha)x_{k}+(1-\alpha)R_{k-1}\le
  x_{k}+(1-\alpha)R_{k-1}<1-\alpha\;.\] Thus,
  $R_{k}<1$. Because we can apply this reasoning to any $x_{k}$ that
  does not trigger the \textbf{if}-condition in the algorithm, we have
  proven the claim.
\end{proof}

And we also need a lemma, analogous to Lemma~\ref{lem:order},
allowing us to reorder the items,

\begin{lemma}\label{lem:order2}
  If $x_1,\ldots,x_n$ is an instance for Algorithm~\ref{fig:Alub}, whose
  packing is triggered by $x_k$, with a competitive ratio larger than
  $\frac{1}{1-\alpha}$, then there is another instance
  $x_{i_1},x_{i_2},x_{i_3},\ldots,x_{i_{k-1}},\allowbreak x_k, \hdots, x_n$ where
  $(i_1,i_2,\ldots,i_{k-1})$ is a permutation of $(1,\ldots,k-1)$ and
  $x_{i_1}\geq x_{i_2}\geq\ldots\geq x_{i_{k-1}}$, that also has a
  competitive ratio larger than $\frac{1}{1-\alpha}$.
\end{lemma}
\begin{proof}
  We show how to transform the original sequence such that the number of
  inversions is reduced by one. Repeating this transformation yields the
  desired result.
  
  Let $x_1,\ldots,x_n$ be an instance for Algorithm~\ref{fig:Alub}, such that
  packing is triggered by $x_k$ with a competitive ratio larger than
  $\frac{1}{1-\alpha}$. 
  Assume that $x_i < x_{i+1}$ for some $i<k-1$. We prove that we
  will get again an instance with a competitive ratio larger than $\frac{1}{1-\alpha}$
  if we swap the positions of $x_i$ and $x_{i+1}$. It is easy to see that
  this transformation preserves the sum of the first $k-1$ items. To show
  that it is a sequence that does not trigger the packing after the $i$th or $(i+1)$st item,
  we have to verify that
  \begin{equation}\label{eq:shorter-sum2}
  x_{i+1}<(1-\alpha)-(1-\alpha)(x_1+\cdots+x_{i-1})
  \end{equation}
  and
  \begin{equation}\label{equ:swap2}
  x_i<(1-\alpha)-(1-\alpha)(x_1+\cdots+x_{i-1}+x_{i+1})\;.
  \end{equation}

  Here, (\ref{eq:shorter-sum2}) is straightforward, since $x_k$ for $k>i+1$ triggering the
  packing in the unmodified instance implies
  $x_{i+1}<(1-\alpha)-(1-\alpha)(x_1+\cdots+x_i)$. We prove
  (\ref{equ:swap2}) by means of a contradiction. Assume, that
  \begin{equation}\label{eq:swap-contradiction2}
  x_i\geq(1-\alpha)-(1-\alpha)(x_1+\cdots+x_{i-1}+x_{i+1})\;.
  \end{equation}
  
  We also know that, for $1\le j\le i+1$, 
  
  \begin{equation}\label{eq:prev-req2}
    x_j\leq (1-\alpha)-(1-\alpha)(x_1-\hdots-x_{j-1})\;,
  \end{equation}
  
  and by construction 
  
  \begin{equation}\label{eq:cond2}
    x_i < x_{i+1}\;.
  \end{equation}
  
  If we define $z=(1-\alpha)-(1-\alpha)(x_1+\cdots+x_{i-1})$, then we can
  rewrite (\ref{eq:swap-contradiction2}), (\ref{eq:prev-req2}) for $j=i+1$, and
  (\ref{eq:cond2}) as
  \begin{align*}
  % x_i,x_{i+1}&\geq0\\
  x_i&\geq z - \textstyle(1-\alpha) x_{i+1},\\
  % x_i&\leq z\\
  x_{i+1}&\leq z-\textstyle(1-\alpha) x_i,\\
  x_i&< x_{i+1}\;.
  \end{align*}
  
  As in the proof of Lemma~\ref{lem:order}, we can solve this system of equations and conclude
  that the first two conditions are only satisfiable if $x_{i+1}\leq x_i$, which gives us the desired contradiction. 
  
  Thus, the new sequence has a competitive ratio larger
  than $\frac{1}{1-\alpha}$. Again, although the order of $x_1,\ldots,x_{k-1}$ has been
  changed, the same reserved requests
  trigger a packing with $x_k$ with the same reservation size, just as before
  the transformation. Hence, if there was no way to pack the items well
  enough before the transformation, it also cannot be done afterwards.
\end{proof}

We are now ready to prove the desired competitive ratio for~Algorithm \ref{fig:Alub}.

\begin{theorem}\label{thm:ublarge}
Algorithm~\ref{fig:Alub} is an online algorithm for \RKP achieving a 
competitive ratio of at most $\frac{1}{1-\alpha}$, for all
$\phi-1\le\alpha< 1$.
\end{theorem}
\begin{proof}
  Let us first assume that we run Algorithm~\ref{fig:Alub} on an
  instance $x_1, \ldots, x_n$, and no element triggers the packing. This
  means that all elements are reserved. But we know by
  Lemma~\ref{lem:res_size_bigalpha} that the reservation never exceeds
  the capacity of the knapsack. This means that the optimal solution
  packs all offered elements. Thus, the algorithm achieves a competitive
  ratio of
  \[\frac{\sum_{i=1}^n x_{k-1}}{\sum_{i=1}^n
    x_{k-1}-\alpha\sum_{i=1}^n x_{k-1}}=\frac{1}{1-\alpha}\;.\] 
    
  Now, it remains to analyze the case where an instance $x_1, \hdots,
  x_n$ triggers the packing, for some $x_k$.  Let us do an induction on
  the value of $k$. If $k=1$, the first item offered triggers the
  packing, thus it holds that $x_1\ge 1-\alpha$ and the gain of the
  algorithm is at least $1-\alpha$ as we expected. Now, we assume that
  Algorithm~\ref{fig:Alub} has a gain of at least $1-\alpha$ on any
  request sequence triggering the algorithm to pack and stop before $k$
  elements are offered.
  
  We proceed similarly to the proof of Theorem~\ref{thm:ubsmall2}.
  Let us assume, using Lemma~\ref{lem:order2}, that we have an instance 
  $x_1\ge x_2 \ge \hdots\ge x_{k-1}, x_k, \hdots x_n$ where
  the algorithm packs after receiving $x_k$. We distinguish two cases.
  \begin{enumerate}
  \item  When the packing is triggered, $x_{k-1}$ is not part of an optimal packing.
  
  In this case, 
  \begin{align}
    \biggl(\sum_{j\text{ is packed}}x_j\biggr)+x_{k-1}>1\;. \label{eq:first}
  \end{align} We consider the following bounds.
  Since $x_{k-1}$ did not trigger the packing, we know that
  \begin{align}
  x_{k-1}+(1-\alpha)\sum_{j<k-1}x_j&<1-\alpha\;, \label{eq:second}
  \end{align}
  and thus also $x_{k-1} < 1-\alpha$.

  Thus, the gain that Algorithm~\ref{fig:Alub} achieves is 
  \begin{align*}
  & \hspace*{-3em} \biggl(\sum_{j\text{ is packed}}x_j\biggr)-\alpha R  \\
  & \geq 1 - x_{k-1} -\alpha \sum_{j\le k-1} x_j &&\text{using (\ref{eq:first})}\\
  & = 1 - (1+\alpha) x_{k-1} - \alpha \sum_{j < k-1} x_j  \\
  & > 1 -(1+\alpha) x_{k-1} -\alpha\biggl(1-\frac{x_{k-1}}{1-\alpha} \biggr) &&\text{using (\ref{eq:second})}\\
  & = 1 - \alpha+\biggl(\frac{\alpha}{1-\alpha}-(1+\alpha)\biggr)x_{k-1} \\
  & \geq 1 - \alpha + \frac{-1+\alpha+\alpha^2}{1-\alpha}x_{k-1} \\
  & \geq 1-\alpha  &&\text{using $\alpha\ge \phi-1$}\;.
  \end{align*}
  Thus we obtain the expected gain in this case.
  
  \item When the packing is triggered, $x_{k-1}$ is part of an optimal packing.
  We can consider two subcases.
  \begin{enumerate}
  \item Taking $x_{k-1}$ out of the request sequence still triggers the packing.
  
  This means that 
  \[x_{k}+(1-\alpha)(R-x_{k-1})\ge 1-\alpha\;\] holds.
  We can thus consider the sequence $x_1,\hdots,x_{k-2},x_{k}$. This sequence has $k-1$ items before the packing is triggered, and $x_{k-1}$ cannot be part of its optimal solution, thus its competitive ratio is at most $1/(1-\alpha)$ by the induction hypothesis, and the competitive ratio after adding $x_{k-1}$ can only get better.
  \item Taking $x_{k-1}$ out of the request sequence does not trigger the packing.
  
  This means that 
  \begin{align}
    x_{k}+(1-\alpha)(R-x_{k-1})< 1-\alpha\;,\label{eq:third}
  \end{align}
  but also
  \begin{align}
    x_{k}+(1-\alpha)R\ge 1-\alpha\;,\label{eq:fourth}
  \end{align}
  
  and, if we let $x_j$ be the smallest element that does not get packed, $x_j\ge x_{k-1}$ holds. Also
  because of the optimality of the packing, $x_{k}\ge x_j$ (otherwise one can take all of the reserved elements as the packing and obtain a better bound) and it holds that
  \begin{align}
    \biggl(\sum_{t\text{ is packed}}x_t\biggr)-x_{k-1}+x_j>1\;.\label{eq:fifth}
  \end{align}
  
  Moreover, we can bound $x_j\le x_1\le 1-\alpha$.
  With these bounds, the gain incurred by the algorithm is at least
  \begin{align*}
  & \hspace*{-3em} \biggl(\sum_{t\text{ is packed}}x_t\biggr)-\alpha R \\
  & \geq 1- x_j +  x_{k-1} -\alpha R &&\text{using (\ref{eq:fifth})}\\
  & \geq 1 -x_j + \frac{x_{k}}{1-\alpha} + R - 1 -\alpha R &&\text{using (\ref{eq:third})}\\
  & = -x_j + \frac{\alpha}{1-\alpha}x_{k}+x_{k} + (1 - \alpha) R \\
  & \geq -x_j + \frac{\alpha}{1-\alpha}x_{k} + (1 - \alpha) &&\text{using (\ref{eq:fourth})}\\
  & = 1-\alpha -x_j + \frac{\alpha}{1-\alpha}x_{k} \\
  & \geq 1-\alpha -x_j + \frac{\alpha}{1-\alpha}x_{j} &&x_k \geq x_j\\
  & = 1-\alpha +\frac{2\alpha -1}{1-\alpha}x_{j}\\
  & \geq 1-\alpha
  \end{align*}
  as we wanted.
  \end{enumerate}
  \end{enumerate}
  This proves the induction step, and thus the desired competitive ratio upper bound.
\end{proof}

We continue by providing matching lower bounds to the upper bounds of
this section.

\section{Lower Bounds}\label{sec:lb}

First we present an adversarial strategy that works for all values
of $\alpha$.  Then we proceed to analyze the case where only
four objects are presented as a generic adversarial strategy and find
improved lower bounds for some values of~$\alpha$.

\subsection{Tight Lower Bound for $0\le \alpha \le 0.25$}
\begin{theorem}\label{lb:smallalpha}
For $\alpha>0$
there exists no algorithm for reservation knapsack
achieving a  competitive ratio better than~$2$.
\end{theorem}
\begin{proof}
Consider the following set of adversarial instances as depicted in Figure~\ref{fig:adv2comp}. Given any
$\epsilon>0$, the adversary presents first an object of size
$\frac{1}{2}+\delta$ with $0<\delta\ll\epsilon$.  If an algorithm
takes this object, an object of size 1 will follow, making its competitive
ratio $1/(\frac12+\delta)>2-\epsilon$.  If an algorithm rejects this
object, no more objects will follow and it will not be competitive.  If an
algorithm reserves this object, then an object of size $\frac12+\delta^2$
will be presented.  Observe, that these two objects do not fit together
into the knapsack.  If an algorithm takes this object, an object of size
$1$ will be presented, and again the algorithm will achieve a competitive
ratio worse than $2-\epsilon$.  If an algorithm rejects this object,
then an object of size $\frac12-\delta^2$ will be presented.  This object
does not fit in the knapsack with the first one, thus the algorithm
can only pack the first object, obtaining a competitive ratio worse
than $2-\epsilon$.  If an algorithm reserves it instead, an object
of size $\frac12+\delta^3$ will be presented.  The adversary can follow
this procedure on and on, and in each step the competitive ratios for
algorithms that accept or reject the offered item only get worse due to
the additional reservation costs.

The adversary can stop offering items as soon as the reservation costs
are such that filling the knapsack will only result in a competitive
ratio worse than~$2$.  This shows that, for every $\epsilon>0$, the
competitive ratio is at least $2-\epsilon$, so the best competitive ratio
is at least~$2$.
\end{proof}

\begin{figure}[tb]
 \centering
 \begin{tikzpicture}[node distance=2cm]
  \node (x1) at (0,0) {$x_1=\frac12+\delta$};
  
  \node[align=left] (x2t) at (0,-2) {$x_2=1$, end\\{\boldmath$\compr_{\alg}> 2-\epsilon$}};
  \node[align=left] (x2rej) at (-2.6,-2) {end\\ {\boldmath$\compr_{\alg}>2$}};
  \node[align=left] (x2res) at (2.8,-2) {$x_2=\frac12+\delta^2$};
  
  \node[align=left] (x3t) at (2.8,-4) {$x_3=1$, end\\{\boldmath$\compr_{\alg}> 2$}};
  \node[align=left,anchor=north east] (x3rej) at (1.2,-3.4) {$x_3=\frac12-\delta^2$, end\\ {\boldmath$\compr_{\alg}=\frac{1}{(1-\alpha)(\frac12+\delta)}> 2$}};
  \node[align=left] (x3res) at (5.6,-4) {$x_3=\frac12+\delta^3$};
  
  \node (x4res) at (7.2,-5) {$\hdots$};
  
  \path[draw] (x1) edge node[left,xshift=-0.5mm] {\tiny{reject}} (x2rej)
  (x1) edge node[right] {\tiny{take}} (x2t)
  (x1) edge node[above right] {\tiny{reserve}} (x2res)
  (x2res) edge node[left,xshift=-1mm] {\tiny{reject}} (x3rej)
  (x2res) edge node[right] {\tiny{take}} (x3t)
  (x2res) edge node[above right] {\tiny{reserve}} (x3res)
  (x3res) edge (x4res);
 \end{tikzpicture}

 \kern-2ex
 \caption{Sketch of the adversarial strategy that is used in the proof of Theorem~\ref{lb:smallalpha}.}
 \label{fig:adv2comp}
\end{figure}
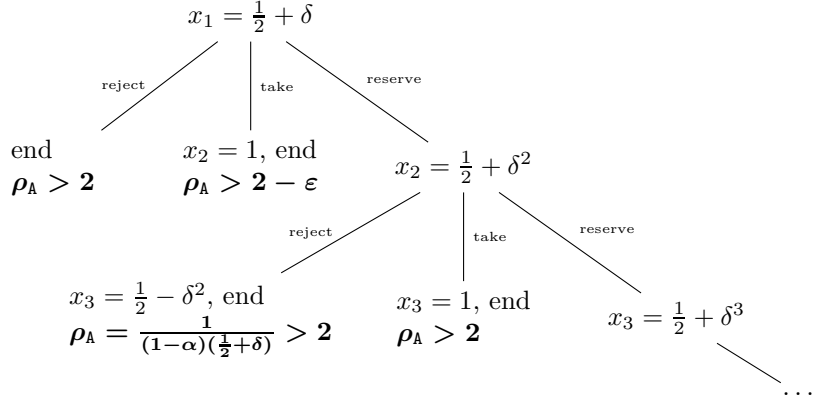
 
The adversarial strategy depicted in Figure \ref{fig:adv2comp} works
for every positive value of $\alpha$, but is of course not tight when
$\alpha>0.25$. Observe that, if an algorithm continues to reserve
items, the cumulated reservation cost will eventually exceed any
remaining gain. 

\subsection{Tight Lower Bound for $0.25\le \alpha \le 1$}

The previous strategy provided us with a lower bound that does not match the
upper bound for $\alpha>0.25$.  We improve the lower bound for larger
$\alpha$ in a way that it provides a matching bound 
by designing a generic adversarial strategy with up to four items as 
shown in Figure~\ref{fig:adv4items}, which contains the competitive ratios
for all possible outcomes. For each item, an algorithm can decide to 
either reject, pack or reserve it. Depending on an algorithm's behavior
on the previous item, it is decided whether another item is presented and, if so, which item. 
We assume by construction that $s<t<u<1$ and $s+t>1$.
The competitive ratio of such a strategy is therefore bounded by
\begin{equation}
\compr_{\alg}\ge \min\Bigl\{\frac{1}{s},~%
\frac{t}{(1-\alpha)s},~\frac{1}{t-\alpha s},~%
\frac{u}{t-\alpha (s+t)},~\frac{1}{u-\alpha(s+t)},~\frac{u}{u-\alpha(s+t+u)} \Bigr\}.
\label{equ:min_main}
\end{equation}
To prove a lower bound for every $0<\alpha<1$, we can choose $s$,
$t$ and $u$ in order to make~(\ref{equ:min_main}) as large as possible.
Standard calculus leads to the bounds in the following theorem.

\begin{figure}
 \centering
 \begin{tikzpicture}[node distance=2cm]
  \node (x1) at (0,0) {$x_1=s$};
  
  \node[align=left] (x2t) at (0,-2) {$x_2=1$, end\\{\boldmath$\compr_{\alg}=\frac1s$}};
  \node[align=left] (x2rej) at (-2.6,-2) {end\\ {\boldmath$\compr_{\alg}=\infty$}};
  \node[align=left] (x2res) at (2.8,-2) {$x_2=t$};
  
  \node[align=left] (x3t) at (2.8,-4) {$x_3=1$, end\\{\boldmath$\compr_{\alg}=\frac{1}{t-\alpha s}$}};
  \node[align=left,anchor=north east] (x3rej) at (1.2,-3.4) {end\\ {\boldmath$\compr_{\alg}=\frac{t}{(1-\alpha)s}$}};
  \node[align=left,anchor=north west] (x3res) at (4.6,-3.4) {$x_4=u$};

  \node[align=left] (x4t) at (5,-6) {$x_4=1$, end\\{\boldmath$\compr_{\alg}=\frac{1}{u-\alpha (s+t)}$}};
  \node[align=left,anchor=north east] (x4rej) at (3.5,-5.5) {end\\ {\boldmath$\compr_{\alg}=\frac{u}{t-\alpha(s+t)}$}};
  \node[align=left,anchor=north west] (x4res) at (7,-5.5) {end\\{\boldmath$\compr_{\alg}=\frac{u}{u-\alpha(s+t+u)}$}};

  \path[draw] (x1) edge node[left,xshift=-0.5mm] {\tiny{reject}} (x2rej)
  (x1) edge node[right] {\tiny{take}} (x2t)
  (x1) edge node[above right] {\tiny{reserve}} (x2res)
  (x2res) edge node[left,xshift=-1mm] {\tiny{reject}} (x3rej)
  (x2res) edge node[right] {\tiny{take}} (x3t)
  (x2res) edge node[above right] {\tiny{reserve}} (x3res)
  (x3res) edge node[left,xshift=-1mm] {\tiny{reject}} (x4rej)
  (x3res) edge node[right] {\tiny{take}} (x4t)
  (x3res) edge node[above right] {\tiny{reserve}} (x4res);
 \end{tikzpicture}
 \caption{Diagram of a generic adversarial strategy with 4 items where $s<t<u<1$ and $s+t>1$.}
 \label{fig:adv4items}
\end{figure}
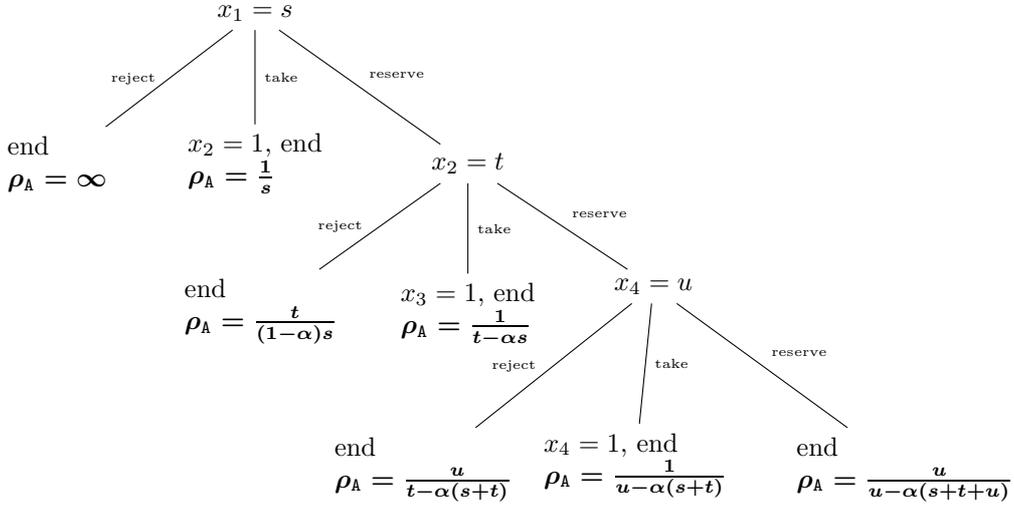

\begin{theorem}\label{lb:smallalphabetter}
The  competitive ratio of \RKP is at least
\begin{alphaenumerate}
\item $\frac{1+\sqrt{5-4\alpha}}{2(1-\alpha)}$ for $0.225\le\alpha\le \sqrt{2}-1$;
\item $2+\alpha$, for $\sqrt{2}-1 \leq \alpha \leq \phi-1$; and
\item $\frac{1}{1-\alpha}$, for $\phi-1 \leq \alpha \leq 1$.
\end{alphaenumerate}
\end{theorem}
\begin{proof}
Let us consider an adversary that presents an item of size $s$
first. The size of this item has to be smaller than $\frac12$, otherwise
an algorithm can take this item and get a competitive ratio smaller
than $2$.  If $s$ is taken, the adversary presents an item of size $1$.
If $s$ is rejected, no more items will be offered.  If $s$ is reserved, an
item of size $t$ is presented, where $s+t>1$.  
If $t$ is taken, the adversary presents an item of
size $1$ and if $t$ is rejected, no more items will be offered,
and, if $t$ is reserved, an item of size $u$ is presented with $u>t$,
if $u$ is taken, the adversary presents an item of size $1$ and otherwise
the request sequence ends in this third item.

For this adversarial strategy, we can compute the competitive
ratio for all possible algorithms as can be seen in Figure~\ref{fig:adv4items}.

Thus, the competitive ratio of any algorithm $\alg$ confronted with
this adversary is bounded by
Equation (\ref{equ:min_main}), which we want to maximize.

Choosing $s=\frac{2}{3+\sqrt{5-4\alpha}}+\varepsilon$, $t=\frac{\sqrt{5-4\alpha}-1+2\alpha}{2(1+\alpha)}$, and
$u=\frac{\alpha+\sqrt{4(t-\alpha)+\alpha^2}}{2}$, we see that, if $0.225\le \alpha \le \sqrt{2}-1$, then $
\compr_{\alg}\ge \frac{1+\sqrt{5-4\alpha}}{2(1-\alpha)}$.

If we compare the given competitive ratio to all of the expressions from (\ref{equ:min_main}), and substitute
the values of $s$, $t$ and $u$ but omit the $\varepsilon$ in the $s$, we see immediately that 
$1/s=\frac{3+\sqrt{5-4\alpha}}{2}\ge \frac{1+\sqrt{5-4\alpha}}{2(1-\alpha)}$ for $\alpha \leq \sqrt{2}-1$, and 
$\frac{1}{t-\alpha s}=\frac{t}{s-\alpha(s+t)}=\frac{1+\sqrt{5-4\alpha}}{2(1-\alpha)}$.

The last three expressions are $\frac{1}{u-\alpha}$, $\frac{u}{t-\alpha}$ and $\frac{u}{u(1-\alpha)-\alpha}$.
For the given $u$, we see that $\frac{1}{u-\alpha}=\frac{u}{t-\alpha}\ge \frac{1+\sqrt{5-4\alpha}}{2(1-\alpha)}$
for $\alpha\ge \alpha_0$, where $\alpha_0\approx 0.224$ is the unique positive root of $x^4+2x^3-2x^2+5x-1$. %0.2240814046756087...
Finally, we observe that $\frac{u}{u(1-\alpha)-\alpha}\geq  \frac{1+\sqrt{5-4\alpha}}{2(1-\alpha)}$ for $\alpha \geq 0.19$
and in particular through the desired range.

Note that $2 \geq \frac{1+\sqrt{5-4\alpha}}{2(1-\alpha)}$ for $\alpha
\leq 0.25$, so while the lower bound is correct for values of $\alpha$
down to approximately $0.224$, our previous lower bound dominates it up to
$\alpha = 0.25$.

For the range of $\sqrt{2}-1\leq \alpha \leq \phi-1$, we choose $s=\frac{1}{2+\alpha}$, $t=1-s+\varepsilon$ and no element $u$.
We can rewrite Equation (\ref{equ:min_main}) as 

\begin{equation}
\compr_{\alg}\ge \min\Bigl\{\frac{1}{s},~%
\frac{1}{t-\alpha s},~\frac{1}{t-\alpha s},~%
\frac{t}{t-\alpha (s+t)} \Bigr\}.
\label{equ:min_2}
\end{equation}
Substituting the appropriate values of $s$ and $t$ in Equation (\ref{equ:min_2}) yields
\[\compr_{\alg}\ge\min\Bigl\{ 2+\alpha, \frac{1+\alpha}{1-\alpha}, \frac{1+\alpha}{1-\alpha-\alpha^2}\Bigr\}.\]
But $\frac{1+\alpha}{1-\alpha}\ge 2+\alpha$ for $\alpha \ge \sqrt{2}-1$ and $\frac{1+\alpha}{1-\alpha-\alpha^2}\ge \frac{1+\alpha}{1-\alpha}$
as long as the denominator is not negative, that is, as long as $\alpha \le \phi -1$.

Finally, for the range $\phi-1 \leq \alpha \leq 1$, we choose $s=1-\alpha$ and no elements $t$ or $u$.
We can rewrite Equation (\ref{equ:min_main}) as 

\begin{equation}
\compr_{\alg}\ge \min\Bigl\{\frac{1}{s},~%
\frac{s}{s-\alpha s} \Bigr\}.
\label{equ:min_3}
\end{equation}
Substituting the appropriate values of $s$ in Equation (\ref{equ:min_3}) yields
$\compr_{\alg}\ge\frac{1}{1-\alpha}$
as expected.
\end{proof}

\section{Nonrejecting Algorithms}\label{sec:nonrej}
A seemingly plausible intuition for \RKP might be
the following: If the cost of reservation is very small, rejecting an item
should not be necessary, as even when an item cannot be packed, the cost of
reserving it is negligible. On the other hand, when the reservation cost is
rising, aggressively reserving items may seem like a very bad strategy, as 
the risk of not being able to utilize reserved items may come to
mind. Interestingly, both of these intuitions turn out to be wrong, which we 
will show by first giving a lower bound for nonrejecting algorithms (i.e., algorithms that reject nothing until they pack the knapsack, after which all remaining items are rejected) that 
exceeds the upper bound of a rejecting algorithm for small $\alpha$ and that 
tightly matches the upper bound of a nonrejecting algorithm for bigger $\alpha$.

We first provide a lower bound
% on the competitive ratio 
for algorithms that are
unable to reject items.

\begin{figure}[tb]
  \centering
  \begin{tikzpicture}[node distance=2cm]
   \node (x1) at (0,0) {$x_1=\frac{1}{2+\alpha}$};
   
   \node[align=left] (x2t) at (-2.8,-1) {$x_2=1$, end\\{\boldmath$\compr_{\alg}=2+\alpha$}};
   \node[align=left] (x2res) at (2.8,-1) {$x_2= \frac{1+\alpha}{2+\alpha} + \epsilon$};
   \node[align=left] (x3t) at (0,-2) {$x_3=1$, end\\{\boldmath$\compr_{\alg}= 2+\alpha$}};
   \node[align=left] (x3res) at (5.6,-2) {$x_3=\frac{1+\alpha}{2+\alpha} + \epsilon$\\{\boldmath$\compr_{\alg}> 2+\alpha$}};
   
   \node (x4res) at (7.2,-3) {$\hdots$};
   \node (x4t) at (2.8,-3) {$\hdots$};

   \path[draw]
   (x1) edge node[above left] {\tiny{take}} (x2t)
   (x1) edge node[above right] {\tiny{reserve}} (x2res)
   (x2res) edge node[above left] {\tiny{take}} (x3t)
   (x2res) edge node[above right] {\tiny{reserve}} (x3res)
   (x3res) edge (x4res)
   (x3res) edge (x4t);
  \end{tikzpicture}

  \kern-2ex
  \caption{Sketch of the adversarial strategy that is used in the proof of Theorem~\ref{thm:nonrejectinglb}. Upon continued reservation, a new item of the same size is presented repeatedly.}
  \label{fig:nonrejecting}
\end{figure}
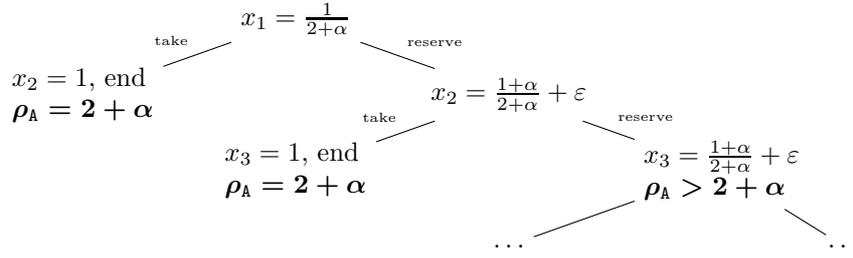

\begin{theorem}\label{thm:nonrejectinglb}
  There exists no deterministic online algorithm for \RKP that does not
  reject any elements with a  competitive ratio better than
  $2+\alpha$ for any $0 < \alpha \leq 1$.
\end{theorem}
\begin{proof}
Let $\epsilon > 0$. Consider the following set of adversarial instances
depicted in Figure~\ref{fig:nonrejecting}. First the adversary presents an
item of size $\frac{1}{2}-\frac{\alpha}{4+2\alpha} = \frac{1}{2+\alpha}$.
If an algorithm takes this item the adversary will present an item of size
1, and the algorithm will have a competitive ratio of $2+\alpha$ as
claimed. If an algorithm reserves this first item, the adversary will
present a second item of size
$\frac{1}{2}+\frac{\alpha}{4+2\alpha}+\epsilon = \frac{1+\alpha}{2+\alpha}
+ \epsilon$. As both items do not fit into the knapsack together, an
algorithm can decide to either take the larger of the two items or to
reserve this second item as well. If it takes the larger item, the
adversary will again present an item of size one. An algorithm can thus
achieve a gain of
\[\frac{1+\alpha}{2+\alpha}+\epsilon - \alpha\frac{1}{2+\alpha} =
\frac{1}{2+\alpha} + \epsilon\] which results in the claimed competitive
ratio. Finally, whenever an algorithm decides to reserve an item, from this
point onwards, another item of size $\frac{1+\alpha}{2+\alpha} + \epsilon$
is presented. At any point at which the algorithm decides to take a
presented item, an item of size 1 is presented afterwards. As there are no
two items in the instance that fit into the knapsack together, the gain of
any deterministic nonrejecting algorithm will strictly decrease with every
reservation.
\end{proof}

Combined with the upper bound given in Theorem~\ref{thm:ubtiny}, we see that an
algorithm that is unable to reject items performs quite a bit worse than one
that is able to reject items, such as Algorithm~\ref{alg:alphatinyupperb}. Thus, an
algorithm needs to be able to reject items to become
2-competitive for small values of $\alpha$.
On the other hand, the lower bound provided in Theorem~\ref{thm:nonrejectinglb}
matches the upper bound of Theorem~\ref{thm:ubsmall}, which is based on the
nonrejecting Algorithm~\ref{fig:Asub}. We conjecture that there are nonrejecting
algorithms for every $\alpha \geq \sqrt{2}-1$ that are at least as good as any
other algorithms that are able to reject items.

\section{Further Work}\label{sec:conclusion} In this work, we provide
tight competitive ratios for the~\RKP. While the model seems natural
to analyze, there are possible variations that could be studied in
order to see how a variation of the reservation model influences the
competitive ratio.  One could consider a variant where reservation
costs are refunded if the item is used, modelling real-world
applications such as a down payment. Another reasonable model could
look at the largest sum of items in a reserve at any time and ask for
reservation costs relative to this size, such as when companies have
to rent storage space.

The concept of reservation may be applied to other online problems
such as online call admission problems in networks
\cite{AwerbuchAP93,0097013,Komm16} or problems of embedding guest
graphs into a host graph.  In the online path packing problem, one
packs paths in a edge-disjoint way (sometimes node-disjointly) into a
graph, which is a generalization of~\RKP, thus inheriting all lower
bounds. The offline version was studied on many types of graphs, with
an incomplete selection being
\cite{Frank90,SchrijverS94,Vanetik09,Bryant10}.
The challenge when applying the idea of reservation is to find an
appropriate cost function that can be measured against the competitive
ratio.

\bibliography{knapsack-journal}

\end{document}